\documentclass{sig-alternate}
\makeatletter
\def\doi#1{\gdef\@doi{#1}}\def\@doi{}
\global\copyrightetc{Copyright \the\copyrtyr\ 
ACM \the\acmcopyr\ ...\$15.00}
\toappear{\the\boilerplate\par{\confname{\the\conf}} \the\confinfo\par \the\copyrightetc.\ifx\@doi\@empty\else\par\@doi.\fi}
\newfont{\mycrnotice}{ptmr8t at 7pt}
\newfont{\myconfname}{ptmri8t at 7pt}
\let\crnotice\mycrnotice%
\let\confname\myconfname%
\makeatother

\usepackage[utf8]{inputenc}
\usepackage[T1]{fontenc}
\usepackage{lmodern}
\usepackage{comment}
\usepackage{amssymb,amsmath}
\usepackage{theorem}
\usepackage{graphics}
\usepackage{amsfonts}
\usepackage{mathrsfs}
\usepackage{amscd}
\usepackage{color}
\usepackage{url}
\usepackage[plainpages=false,pdfpagelabels,colorlinks=true,citecolor=blue,hypertexnames=false]{hyperref}
\usepackage{alltt}
\usepackage{bbm}

\usepackage{microtype}  
\usepackage{flushend}   
\usepackage{mathtools}  
\usepackage{booktabs}   
\usepackage{array}      

\usepackage{mathptmx}
\newcommand{\A}{\mathbf{a}}
\renewcommand\Re{\operatorname{Re}}

\newcommand\Sp{\operatorname{Sp}}
\newcommand\res{\operatorname{res}}
\newcommand\numer{\operatorname{numer}}

\newtheorem{theorem}{Theorem}
\newtheorem{prop}{Proposition}
\newtheorem{cor}{Corollary}
\newtheorem{lem}{Lemma}

\newtheorem{definition}{Definition}
\newenvironment{example}[1][Example]{\begin{trivlist} \item[\hskip \labelsep {\bfseries #1}]}{\end{trivlist}}

\usepackage{color}

\def\C {\ensuremath{\mathbb{C}}}

\def\Q {\ensuremath{\mathbb{Q}}}

\def\N {\ensuremath{\mathbb{N}}}
\def\z {\ensuremath{\mathbf{z}}}
\def\p {\ensuremath{\mathbf{p}}}
\def\q {\ensuremath{\mathbf{q}}}
\def\a {\ensuremath{\mathbf{a}}}
\def\b {\ensuremath{\mathbf{b}}}
\def\c {\ensuremath{\mathbf{c}}}
\def\d {\ensuremath{\mathbf{d}}}
\def\e {\ensuremath{\mathbf{e}}}
\def\f {\ensuremath{\mathbf{f}}}
\def\r {\ensuremath{\mathbf{r}}}
\def\T {\ensuremath{\mbox{\boldmath$\theta$}}}
\def\balpha {\ensuremath{\mbox{\boldmath$\alpha$}}}
\def\bbeta {\ensuremath{\mbox{\boldmath$\beta$}}}
\def\bgamma {\ensuremath{\mbox{\boldmath$\gamma$}}}
\def\bdelta {\ensuremath{\mbox{\boldmath$\delta$}}}

\newenvironment{my_enumerate}{
\vspace{-6pt}
\begin{enumerate}
     \setlength{\itemsep}{0pt}
     \setlength{\parskip}{0pt}
     \setlength{\parsep}{0pt}}
{\end{enumerate}
\vspace{-6pt}
}

\usepackage{times}
\begin{document}
\permission{Permission to make digital or hard copies of all or part of this work for personal or classroom use is granted without fee provided that copies are not made or distributed for profit or commercial advantage and that copies bear this notice and the full citation on the first page. Copyrights for components of this work owned by others than ACM must be honored. Abstracting with credit is permitted. To copy otherwise, or republish, to post on servers or to redistribute to lists, requires prior specific permission and/or a fee. Request permissions from Permissions@acm.org. }

\conferenceinfo{ISSAC '14,}{July 21 - 25, 2014, Kobe, Japan.}
\CopyrightYear{2014}
\crdata{978-1-4503-2501-1/14/07}
\doi{http://dx.doi.org/10.1145/2608628.2608662}
	

\setlength{\abovedisplayskip}{.36\baselineskip} \setlength{\abovedisplayshortskip}{0pt}

\title{Computing necessary integrability conditions\\ for planar parametrized homogeneous potentials\titlenote{We thank the referees for their helpful comments. This work has been supported in part by the Microsoft Research\,--\,Inria Joint Centre.}\vspace{-25pt}}
\newfont{\authfntsmall}{phvr at 11pt}
\newfont{\eaddfntsmall}{phvr at 9pt}

\numberofauthors{3} 
\author{
\alignauthor
Alin Bostan\\
\affaddr{INRIA (France)}\\
\email{\textsf{\normalsize{Alin.Bostan@inria.fr}}}
\alignauthor
Thierry Combot\\
\affaddr{Institut de Mathématiques de Bourgogne
  UMR CNRS 5584}\\
\affaddr{Univ. de Bourgogne (France)}\\
\email{\textsf{\normalsize{thierry.combot@u-bourgogne.fr}}}
\alignauthor
Mohab {Safey El Din}\\
\affaddr{Sorbonne Universities} \\
\affaddr{Univ. Pierre et Marie Curie (Paris 06)}\\
\affaddr{INRIA Paris Rocquencourt, POLSYS Project}\\
\affaddr{LIP6 CNRS, UMR 7606}\\
\affaddr{Institut Universitaire de France}\\
\email{\textsf{\normalsize{Mohab.Safey@lip6.fr}}}	
}
\date{today}

\maketitle
\begin{abstract} 
Let $V\in\mathbb{Q}(i)(\a_1,\dots,\a_n)(\q_1,\q_2)$ be a rationally parametrized planar homogeneous potential of homogeneity degree $k\neq -2, 0, 2$. We design an algorithm that computes polynomial \emph{necessary} conditions on the parameters $(\a_1,\dots,\a_n)$ such that the dynamical system associated to the potential $V$ is integrable. These conditions originate from those of the Morales-Ramis-Sim\'o integrability criterion near all Darboux points. The implementation of the algorithm allows to treat applications that were out of reach before, for instance concerning the non-integrability of polynomial potentials up to degree $9$. Another striking application is the first complete proof of the non-integrability of the \emph{collinear three body problem}. 	
\end{abstract}

\vspace{12mm}
\noindent
{\bf Categories and Subject Descriptors:} \\
\noindent I.1.2 [{\bf Computing Methodologies}]: Symbolic and
Algebraic Manipulations --- \emph{Algebraic Algorithms}\\

\vspace{7mm}
\noindent {\bf General Terms:} Algorithms, Theory.\\

\vspace{7mm}

\section{Introduction}\label{sec:intro}

Let us consider the Hamiltonian system
\begin{equation}\label{eq:H}
\dot{\q}_1=\p_1,\quad \dot{\q}_2=\p_2,\qquad \dot{\p}_1=-\frac{\partial V}{\partial \q_1 },\qquad \dot{\p}_2=-\frac{\partial V}{\partial \q_2}
\end{equation}
with $V\in\mathbb{C}(\q_1,\q_2)$, called the
\emph{potential}.
System~\eqref{eq:H} describes the motion of a particle in
the plane submitted to the force field $\nabla V(\q)$. It always admits the
so-called \emph{Hamiltonian\/}
$H=\textstyle{\frac{1}{2}}({\p}_1^2+{\p}_2^2)+V(\q_1,\q_2)$ as a rational 
first
integral. The potential~$V$ is called \emph{(rationally) integrable\/} 
if system~\eqref{eq:H} admits another rational first integral~$I$,
functionally independent on~$H$. 
Intuitively, the integrability of~$V$ is equivalent to the fact that~\eqref{eq:H} can be solved in explicit terms.

Integrability is a rare phenomenon and it is in general a difficult
task to determine whether a given potential is integrable or not. For
\emph{homogeneous potentials} in~$\mathbb{C}(\q_1,\q_2)$,
\emph{necessary\/} conditions for integrability were given by
Morales-Ramis~\cite{MoRa01a} and by
Morales-Ramis-Sim\'o~\cite{MoRaSi07}.  Building on these works, we
design in this article an algorithm which takes as input a
\emph{family\/} of rational homogeneous potentials
$V\in\mathbb{Q}(i)(\mathbf{a})(\q_1,\q_2)$ depending on parameters
${\mathbf a} = (\a_1,\ldots,\a_n)$ and which computes a set of
constraints on the parameter values $a\in\mathbb{C}^n$ that are
necessary for the integrability of~$V(a, \q_1, \q_2)$. These
constraints turn out to be of polynomial nature in $\a$.

There are several difficulties in this parameterized setting. The first one is
that the integrability constraints provided by the Morales-Ramis theory ---on
which our whole approach relies---, are expressed in terms of quantities
(eigenvalues of Hessian matrices at Darboux points, see
Section~\ref{sec:preliminaries}) which are not easily accessible. We circumvent this basic difficulty by using an equation that relates the
eigenvalues, but this brings a new
technical complication since the equation is of Diophantine type. A third difficulty is that the number of Darboux
points itself may depend on the parameters, leading to \emph{singular} cases.

We follow a classical approach, inspired mostly by ideas in~\cite{MaPr05}. Our
contribution to the topic is effective and algorithmic, as we provide a
complete, proven and implemented algorithm for the problem of computing
necessary integrability conditions for planar parametrized homogeneous
potentials, with precise output specifications. Our algorithm uses classical
tools in computer algebra, such as polynomial ideal elimination based on
Gr\"{o}bner bases techniques. An important feature is the use of (complex)
polar coordinates to represent homogeneous potentials by univariate rational
functions with parameters $F\in \mathbb{Q}(i)(\mathbf{a})(\z)$. This change of
representation considerably simplifies the computations and the proofs. For
instance, in polar representation, \emph{singular\/} cases are those with
non-generic multiplicity of the roots/poles of $F$. They are treated by our
algorithm, which builds a tree containing each possible singular case. This
approach is related with comprehensive Gr\"{o}bner
bases~\cite{Weispfenning92}, which are avoided here thanks to some a priori
knowledge about singular cases.

In summary, our strategy for computing necessary integrability conditions for~$V$ consists in~4 steps:
\emph{(i)} rewrite $V$ in polar coordinates; \emph{(ii)} set up a Diophantine
equation whose solutions belong to the so-called \emph{Morales-Ramis
table} (that contains all possible eigenvalues of the Hessian of $V$ at Darboux
points of $V$); \emph{(iii)} solve this Diophantine equation; \emph{(iv)}
rewrite the condition of having prescribed eigenvalues at Darboux points as
polynomial conditions on~$\mathbf{a}$.

Some prior works  used a similar strategy, but it was unclear which cases were possible to tackle, in particular for singular ones. The approach was not fully automatized and this explains that results were
only available for special families of potentials, for instance polynomials of small degree
(3 or 4)~\cite{MaPr04,MaPr05,LiMaVa11,LiMaVa11b}, as the number of singular cases grows very fast (already $44$ for polynomials of degree~$5$).
By contrast, our treatment is unified and fully automated, and it
allows not only to retrieve (and sometimes correct) known results, but more
importantly, to treat potentials of degrees previously unreached (up to~9).
By applying our algorithm to polynomial potentials, we found three new
cases admissible for integrability at degree~$5$ (but still not proved
to be integrable), and various new families for higher degrees.  An
even more striking application of our algorithm is the first complete
proof of the non-integrability of the \emph{collinear three body
  problem}, on which only partial results were
known~\cite{Yoshida87,MoSi09,Shibayama11}.  {The direct
approach that consists in searching first integrals
\cite{Hietarinta83,NaYo01} is complementary to our (non-)integrability
analysis, as our algorithm helps either proving that the lists in
\cite{Hietarinta83,NaYo01} are complete, or finding new unknown cases.}

\medskip
{\emph{Warning: We will assume throughout the article that the homogeneity degree $k$ is different from $-2,0$ and $2$}}. 
(This is because the Morales-Ramis theory is much less powerful when $k\in \{-2,0,2\}$.)

\smallskip\emph{Convention of notation: to avoid confusion, we will use bold letters for variables/parameters, and italic letters for parameter values}.

\vspace{-0.15cm}
\section{Preliminaries and basic notions}
\label{sec:preliminaries}

There exist strong integrability constraints (see Theorem~\ref{thm:morales} below). They require to deal with Darboux points, whose definition we now recall.

\begin{definition}\label{darb}
  Let $V\in\mathbb{C}(\q_1,\q_2)$ be a homogeneous rational function
  of homogeneity degree $k\neq 0$. A point $c=(c_1,c_2)\in\mathbb{C}^2\setminus
 \{0\}$ is called a \emph{(proper) Darboux point} of $V$ if it
  satisfies the equations
\begin{equation}\label{eqdarb}
\frac{\partial V}{\partial \q_1 }(\c)=k \c_1, \quad \frac{\partial V}{\partial \q_2 }(\c)=k \c_2.
\end{equation}
\end{definition} 

\noindent {Note that, by homogeneity, we could have chosen an arbitrary
normalization non-zero constant on the right-hand side of~\eqref{eqdarb}. In
the literature, this normalization constant is frequently chosen equal
to~$1$~\cite{MoRa01}. However, our choice is deliberate, see the remark after
Theorem~\ref{thm:morales}.}

The following result (which is an application of a more general criterion due
to Morales and Ramis~\cite{MoRa01a}) provides \emph{necessary} conditions for
integrability under the form of constraints on eigenvalues of Hessian matrices
at each Darboux point.
It is the basic ingredient for numerous non-integrability
proofs~\cite{MaPr04,MaPr05,MoSi09,MoRa01b,Tsygvintsev01,MaPrYo08,MaPrYo12,Przybylska07,Boucher00}.
{Roughly, its main idea is as follows. A Darboux point
leads to a straight line orbit of the dynamical system~\eqref{eq:H} associated
to~$V$, around which the system~\eqref{eq:H} can be linearized. If the whole
system is integrable, then the linearized system, which in our case
corresponds to a hypergeometric equation, is also integrable. Thus the
integrability table of Theorem~\ref{thm:morales} below is reminiscent of
Kimura's classification \cite{Kimura70} of solvable hypergeometric equations.}

\begin{theorem}\label{thm:morales} (Morales-Ramis \cite{MoRa01})
  Let $V\in\mathbb{C}(\q_1,\q_2)$ be a homogeneous rational function
  of homogeneity degree $k\neq -2, 0, 2$, and let $c\in\mathbb{C}^2\setminus
  \{0\}$ be a Darboux
  point of $V$. If the potential $V$ is integrable, then for any
  eigenvalue~$\lambda$ of the Hessian matrix of $V$ at $c$, the pair
  $(k,\lambda)$ belongs to the following table, for some
  $j\in\mathbb{Z}$.
{\small
\begin{center} 
{\small\renewcommand{\arraystretch}{1.3} \tabcolsep6pt
\begin{tabular}{|c|c||c|c|}\hline
$k$&$\lambda$&$k$&$\lambda$\\\hline
$\mathbb{Z}^*$&$\frac{1}{2}jk(jk+k-2)$& $-3$&$-\frac{25}{8}+\frac{1}{8}(\frac{12}{5}+6 j)^2$ \\\hline
$\mathbb{Z}^*$&$\frac{1}{2}(jk+1)(jk+k-1)$&$3$&$-\frac{1}{8}+\frac{1}{8}(2+6 j)^2$ \\\hline
$-5$&$-\frac{49}{8}+\frac{1}{8}(\frac{10}{3}+10 j)^2$                 &$3$&$-\frac{1}{8}+\frac{1}{8}(\frac{3}{2}+6 j)^2$ \\\hline
$-5$&$-\frac{49}{8}+\frac{1}{8}(4+10 j)^2$                 & $3$&$-\frac{1}{8}+\frac{1}{8}(\frac{6}{5}+6 j)^2$ \\\hline
$-4$&$-\frac{9}{2}+\frac{1}{2}(\frac{4}{3}+4j)^2$	& $3$&$-\frac{1}{8}+\frac{1}{8}(\frac{12}{5}+6 j)^2$ \\\hline
$-3$&$-\frac{25}{8}+\frac{1}{8}(2+6 j)^2$	& $4$&$-\frac{1}{2}+\frac{1}{2}(\frac{4}{3}+4 j)^2$ \\\hline
$-3$&$-\frac{25}{8}+\frac{1}{8}(\frac{3}{2}+6 j)^2$	& $5$&$-\frac{9}{8}+\frac{1}{8}(\frac{10}{3}+10 j)^2$ \\\hline
$-3$&$-\frac{25}{8}+\frac{1}{8}(\frac{6}{5}+6 j)^2$	& $5$&$-\frac{9}{8}+\frac{1}{8}(4+6 j)^2$ \\\hline
\end{tabular}}
\end{center}}
\end{theorem}

This table will be called throughout the article the \emph{Morales-Ramis table}.
For a fixed homogeneity degree $k$, we will denote by $E_k$ the infinite set of \emph{allowed eigenvalues $\lambda$\/} in the table, corresponding to $k$.

\smallskip
Note two differences with the classical statement of the Morales-Ramis
theorem. First, due to our choice of the normalization constant in
Definition~\ref{darb} ($k$ instead of $1$), the eigenvalues
displayed in the previous table are $k$ times larger than those
of~\cite[Theorem~3]{MoRa01}.  Our choice is motivated by the fact that
it simplifies the computations, and it has the nice and useful
property that the eigenvalue sets in the table are lower bounded. 
Second, both the original proof~\cite{MoRa01} and the statement of the
Morales-Ramis theorem~\cite[Theorem~1.2]{DuMa09}, require the
additional assumption that the Hessian matrix of $V$ at $c$ is
diagonalizable; but in fact, \cite[Theorem~1.3(1)]{DuMa09} shows that this hypothesis is not necessary. 

\medskip
We now illustrate the basic notion of Darboux points and the use of Theorem~\ref{thm:morales} on a toy parametrized example. The example is simple enough so that the eigenvalues are accessible by a direct computation. 

\begin{example} Consider the homogeneous potential
\begin{equation*}\label{eqex}
V(\A, \q_1, \q_2)=(\a_1 \q_1+\a_2\q_2)(\q_1^2+\q_2^2).
\end{equation*}
The homogeneity degree is $k=3$ and the Darboux point equation~\eqref{eqdarb}~is
\[
3\a_1\c_1^2+\a_1\c_2^2+2\a_2\c_1\c_2=3\c_1, \;\; 3\a_2\c_2^2+\a_2\c_1^2+2\a_1\c_1\c_2=3\c_2.
\]
For parameter values $(a_1, a_2)\in \mathbb{C}^2$ such that $a_1^2+a_2^2\neq 0$, its solutions $c=(c_1,c_2)\in\mathbb{C}^2\setminus \{0\}$ read 
$$c=\left(\dfrac{a_1}{a_1^2+a_2^2},\; \dfrac{a_2}{a_1^2+a_2^2}\right),\quad c = \left(\dfrac{3}{2(a_1\pm ia_2)},\dfrac{\pm 3i}{2(a_1\pm ia_2)}\right).$$
 The Hessian matrices at these points are
$$\begin{pmatrix} \dfrac{2 (3 a_1^2+a_2^2)}{a_1^2+a_2^2}&\hspace{-0.3cm}\dfrac{4a_1a_2}{a_1^2+a_2^2}\\& \\ \hspace{-0.5cm}\dfrac{4a_1a_2}{a_1^2+a_2^2}&\hspace{-0.3cm}\dfrac{2(a_1^2+3a_2^2)}{a_1^2+a_2^2}\\\end{pmatrix},\; 
\begin{pmatrix} \dfrac{3(3a_1 \pm ia_2)}{a_1\pm ia_2}&\hspace{-0.2cm}\dfrac{3(a_2 \pm ia_1)}{a_1\pm ia_2}\\ & \\ \hspace{-0.3cm} \dfrac{3(a_2 \pm ia_1)}{a_1\pm ia_2}&\hspace{-0.2cm}\dfrac{3(a_1 \pm 3ia_2)}{a_1\pm ia_2}\\\end{pmatrix}.
$$

\noindent The eigenvalues of the first matrix are $\{6,2\}$. The (a priori unexpected) fact that none of them depend on the parameter values $a_1,a_2$ comes from a relation on eigenvalues at Darboux points that will be proved later (Theorem \ref{Macie}). Now, Theorem~\ref{thm:morales} tells us that $E_3$, the set of  allowed eigenvalues for homogeneity degree $k=3$, is the set of the rational numbers of the form
\begin{align*}
\frac{3}{2}j(3j+1),\; \frac{1}{2}(3j+1)(3j+2), \;-\frac{1}{8}+\frac{1}{8}\left(6 j+\frac{12}{5}\right)^2
,\\ -\frac{1}{8}+\frac{1}{8}\left(6 j + \dfrac32\right)^2,
-\frac{1}{8}+\frac{1}{8}\left(6 j+ \frac{6}{5}\right)^2, -\frac{1}{8}+\frac{1}{8}(6 j+2)^2,
\end{align*}
where $j \in \mathbb{Z}$.
The eigenvalue $\lambda=6$ is allowed (by choosing $j=1$ in the first sequence), but the eigenvalue $\lambda=2$ is not. This can be seen by simply solving for integers six quadratic equations.
Thus, by Theorem~\ref{thm:morales}, the potential $V(a_1, a_2, \q_1, \q_2)$  is not
integrable when $a_1^2+a_2^2\neq 0$, and a necessary condition for integrability is $a_1^2+a_2^2=0$. 
\end{example}


\section{Polar representation}\label{sec:polar}

We will use complex polar coordinates in order to represent a given
rational homogeneous potential $V$ in a simpler way, by a pair
$(F,k)$, where $F$ is a univariate rational function, and $k$ is an
integer. In this new representation, various quantities
attached to $V$, such as Darboux points and eigenvalues of the Hessian
of $V$, are much easier to express, including a useful relation
\eqref{eqMac} on these eigenvalues.  This representation has already
been used for non-integrability proofs~\cite{SVF97,MOY10,VSF03}. This
section provides an overview on some results on polar coordinates with
useful properties needed to prove our algorithm (see Theorem
\ref{thm:morales_polar} below).

In the rest of the article, we will use the notation $\Delta$ and $D$ for the following subdomains of $\mathbb{C}^2$:
\begin{align*}
\Delta & =   \mathbb{C}^*\times \{\theta\in\mathbb{C},\, \; 0\leq \Re(\theta)<2\pi\}, \\ D & =  \{(q_1,q_2)\in\mathbb{C}^2, \; q_1^2+q_2^2\neq 0\},
\end{align*}
and $\varphi$ for the map $\varphi:\Delta\rightarrow D$ defined by $\varphi(r,\theta) = (r\cos\theta,r\sin\theta)$.

\begin{prop}\label{prop:polar-coord}
The map $\varphi$ is differentiable on $\Delta$, and its image is a double covering of $D$ (i.e., each fiber $\varphi^{-1}(q_1,q_2)$ contains exactly two points).
\end{prop}

\begin{proof}
The functions $(r, \theta) \mapsto r\cos\theta$ and $ (r, \theta) \mapsto r\sin\theta$ are differentiable on $\Delta$, and thus $\varphi$ is differentiable on $\Delta$. The relation $(r\cos\theta)^2+(r\sin\theta)^2=r^2$ implies that the image of $\varphi$ is contained in~$D$. Let us compute the inverse of~$\varphi$. If $q_1 = r\cos\theta$ and $q_2 = r\sin\theta$ with $(r,\theta)\in \Delta$, then
$r^2=q_1^2+q_2^2$ and $e^{i\theta}=(q_1+iq_2)/r.$
The first relation determines $r$ up to a sign; since
$q_1^2+q_2^2\neq 0$ there are always exactly two possible choices $\pm
r$.  After this sign choice, $e^{i\theta}$ is uniquely
determined, thus $\theta$ is determined up to translation by
$2\pi$. Since  $0\leq \Re(\theta)<2\pi$, then
$\theta$ is uniquely determined.
\end{proof}

\begin{prop}\label{prop:polarV}
Any homogeneous potential $V\in\mathbb{C}(\q_1,\q_2)$ can be written in complex polar coordinates
\begin{equation}\label{represent}
V(q_1,q_2) = r^k F(e^{i\theta}), \qquad \text{for} \quad (q_1,q_2) = \varphi(r,\theta),
\end{equation}
where $k$ is the homogeneity degree of $V$, and $F$ is a rational function in $\mathbb{C}(\z)$ having the same parity as $k$.

Moreover, if $V\in\mathbb{Q}(\q_1,\q_2)$, then $F$ belongs to $\mathbb{Q}(i)(\z)$.
\end{prop}

\begin{proof}
Let $F$ be the rational function in $\mathbb{C}(\z)$ defined by
$$F(\z)=V\left(\frac{\z+\z^{-1}}{2},\frac{\z-\z^{-1}}{2i}\right).$$
Then $F(e^{i\theta})$ is equal to $V(\cos\theta,\sin\theta)$, and homogeneity of $V$ allows to conclude that for $q_1 = r \cos(\theta)$ and $q_2 = r \sin(\theta)$ we have 
\[V(q_1,q_2) = r^k V(\cos\theta,\sin\theta)= r^k F(e^{i\theta}).\]
 Using again that $V$ is $k$-homogeneous, we obtain that
\begin{align*}
F(-\z)= \, V\left(-\frac{\z+\z^{-1}}{2},-\frac{\z-\z^{-1}}{2i}\right)
= \, (-1)^kF(\z),
\end{align*}
and thus $F$ has the same parity as $k$. 
The last assertion is obvious by definition of $F$.
\end{proof}

Proposition~\ref{prop:polarV} shows that the homogeneous rational potential $V$ is represented in polar coordinates by a pair $(F,k)$, where $F$ is a univariate rational function, and $k$ is an integer. 
We now write the equation of a Darboux point~$c\in D$ of a potential $V$ and the eigenvalues of the corresponding Hessian matrix $\nabla^2 V(c)$ in polar coordinates.

\begin{prop}\label{prop:spectrum}
  Let $V\in\mathbb{C}(\q_1,\q_2)$ be a homogeneous potential with
  polar representation $(F,k)$, and let $c=(c_1,c_2)\in D$ be a 
  Darboux point for $V$. Then for $(r,\theta) \in
  \varphi^{-1}(c)$ we have
\begin{equation}\label{eqdarb2}
F'(e^{i\theta})=0\qquad \text{and} \qquad F(e^{i\theta})=r^{2-k}.
\end{equation}
Moreover, $(c_1,c_2)^t$ and $(-c_2,c_1)^t$ are eigenvectors of the Hessian matrix $\nabla^2 V(c)$, with respective eigenvalues
\begin{equation}\label{eq:SpecHessian}
 k(k-1)\qquad \text{and} \qquad k-\frac{e^{2i\theta}F''(e^{i\theta})}{F(e^{i\theta})}.
\end{equation}
\end{prop}

\begin{proof}
We start from the relations
\[V(\q_1,\q_2) = \r^k \, F(\z), \quad \r^2 = {\q_1^2 + \q_2^2}, \quad \z = \frac{\q_1 + i \q_2}{\r} = e^{i \T}. \]
{}From there we deduce, by differentiation, the equalities 
\[ \frac{\partial \r}{\partial \q_1} = \frac{\q_1}{\r}, \quad \frac{\partial  \r}{\partial \q_2} = \frac{\q_2}{\r}, \quad \frac{\partial \z}{\partial \q_1} = \frac{-i\q_2\z}{\r^2}, \quad \frac{\partial \z}{\partial \q_2} = \frac{i\q_1\z}{\r^2}.\] 
These equalities imply that the derivatives of $V$ write
\begin{align}
\frac{\partial V}{\partial \q_1 } = \r^{k-2} \left( k \q_1 F(\z)-i \q_2 \z F'(\z) \right),\label{eq:derV1}\\
\frac{\partial V}{\partial \q_2 } = \r^{k-2} \left( k \q_2 F(\z)+i \q_1 \z F'(\z) \right). \label{eq:derV2}
\end{align}
Combining these last two equations yields
\begin{equation}\label{eq:Euler}
\q_1 \frac{\partial V}{\partial \q_1 } + \q_2 \frac{\partial V}{\partial \q_2 } = k V, \;\;\;
\q_1 \frac{\partial V}{\partial \q_2 } - \q_2 \frac{\partial V}{\partial \q_1 } = i \, \r^{k} \, \z \, F'(\z). 
\end{equation}
(The first one is Euler's relation for $k$-homogeneous functions.) 
Evaluating equalities~\eqref{eq:Euler} at the Darboux point~$c$, and using the  Darboux point equation \eqref{eqdarb}, yields the proof of~\eqref{eqdarb2}. 

Let us now prove the last assertion of the proposition.
By differentiating the first equality in~\eqref{eq:Euler} with respect to~$\q_1$ and~$\q_2$, by evaluating at~$c$, and by using~\eqref{eqdarb2}, we obtain $\nabla^2V(c).c^t+kc^t=k^2c^t$. Thus $c^t$ is an eigenvector of $\nabla^2V(c)$, with corresponding eigenvalue $k(k-1)$.

Similarly, differentiating the second equality in~\eqref{eq:Euler} and specializing the result at $c$ yields $\nabla^2V(c).v -kv=-r^{k-2}e^{2i\theta} F'' (e^{i\theta}).v$, where $v$ denotes the vector $(-c_2,c_1)^t$.
This concludes the proof.
\end{proof}

Proposition \ref{prop:spectrum} motivates the following definition of
\emph{Darboux points in polar representation}, and of \emph{associated eigenvalues}.
\vspace{-0.5em}

\begin{definition}\label{darbF}
  Let $(F,k)$ be the polar representation of a homogeneous potential
  $V\in \mathbb{C}(\q_1,\q_2).$ A complex number $z\in
  \mathbb{C}\setminus \{0\}$ is called a \emph{Darboux point of $F$} if
  $F'(z)=0$ and $F(z)\neq 0$.
A Darboux 
point $z\in
\mathbb{C}\setminus \{0\}$ is said to be \emph{multiple} if $z$ is a
multiple root of $F'$; else it is said to be \emph{simple}.

If $z\in \mathbb{C}\setminus \{0\}$ is a Darboux point
for $F$, we will call \emph{associated eigenvalues} the values
$k(k-1)$ and $ k-{z^2 F''(z)}/{F(z)}$.
\end{definition}

The map $\varphi$ naturally sends Darboux points in polar representation to Darboux points in Cartesian coordinates in $D$, also carrying the definition of \emph{associated eigenvalues}.

We now prove the main result of this subsection; it gives a necessary
condition for integrability of a homogeneous potential in terms of its
polar representation. We recall that $E_k$ is the set of allowed
values in the Morales-Ramis table for degree $k$.

\begin{theorem}\label{thm:morales_polar}
  Let $V\in \mathbb{C}(\q_1,\q_2)$ be a homogeneous potential with polar
  representation $(F,k)$ and let $\Lambda$ be the following set
\begin{equation}\label{Lambda:def}
\Lambda(F,k) := \left\lbrace k-\frac{z^2F''(z)}{F(z)} \quad \Big| \quad z\neq 0,\;F'(z)=0,\; F(z)\neq 0\right\rbrace.
\end{equation}
Let $\Sp_D(\nabla^2 V)$ denote the union of the sets $\Sp(\nabla^2 V(c))$ taken over all Darboux points $c\in D$ of $V$. Then 
\begin{equation}\label{eq:SpD-Lambda}
\{k(k-1)\} \cup \Sp_D(\nabla^2 V) = \{k(k-1)\} \cup \Lambda.
\end{equation}
Moreover, if $V$ is integrable, then $\Lambda \subseteq E_k$.
\end{theorem}

\begin{proof}
  We first prove equality \eqref{eq:SpD-Lambda}.
  Proposition~\ref{prop:spectrum} readily yields the inclusion
  $\{k(k-1)\} \cup \Sp_D(\nabla^2 V) \subseteq \{k(k-1)\} \cup
  \Lambda.$ Indeed, if $\lambda$ is in $\Sp_D(\nabla^2 V)\setminus
  \{k(k-1)\}$, then there exists a Darboux point $c\in
  D$ of $V$ such that $\lambda \in \Sp(\nabla^2 V(c))\setminus
  \{k(k-1)\}$. Then letting $(r,\theta) \in \varphi^{-1}(c)$,
  Proposition~\ref{prop:spectrum} implies that $z=e^{i\theta}$
  satisfies $z\neq 0$, $F'(z)=0$, $F(z)\neq 0$ and $\Sp(\nabla^2 V(c))
  = \{ k(k-1), k - z^2 F''(z)/F(z)\}$. Therefore, $\lambda$ is equal
  to $k - z^2 F''(z)/F(z)$, and thus it belongs to $\Lambda$.

  Conversely, let $\lambda$ be in $\Lambda \setminus \{k(k-1)\}$.
  There exists a $z \in \mathbb{C}\setminus \{0\}$ such that
  $F'(z)=0$, $F(z)\neq 0$ and $\lambda = k - z^2
  F''(z)/F(z)$.  Write this $z$ as $e^{i \theta}$ for some
  $\theta$ with $\Re(\theta) \in [0,2\pi)$, and write $F(z)$ as
  $r^{2-k}$ with $r \in \mathbb{C}^*$. Then,
  Equations~\eqref{eq:derV1} and \eqref{eq:derV2} imply that $c=
  \varphi(r,\theta)$ is a Darboux point of $V$ in
  $D$. By Proposition~\ref{prop:spectrum}, $\lambda$ belongs to
  $\Sp_D(\nabla^2 V)$. Equality \eqref{eq:SpD-Lambda} is now proven.
		
  To prove the last assertion, assume that $V$ is integrable. Then
  Theorem~\ref{thm:morales} shows that $\Sp_D(\nabla^2 V) \subseteq
  E_k$.  Since the eigenvalue $k(k-1)$ belongs to the Morales-Ramis
  table (first sequence with $j=1$), we also have $\{k(k-1)\} \cup
  \Sp_D(\nabla^2 V) \subseteq E_k$.  The desired inclusion $\Lambda
  \subseteq E_k$ is then a consequence of
  equality~\eqref{eq:SpD-Lambda}.
\end{proof}	

\section{A special Diophantine equation}

\subsection{Sets of possible eigenvalues}\label{secSet1}

There are infinitely many allowed eigenvalues for integrability in the
Morales-Ramis table.  We now prove an interesting equation that
relates the eigenvalues of the Hessian matrix of $V$ at Darboux
points. This will allow us to bound the possible eigenvalues allowed
for integrability.

\begin{definition}
  For a rational function $F\in\mathbb{C}(\z) \setminus \{0\}$, we denote by $k_0=k_0(F)$
  and $k_\infty=k_\infty(F) \in \mathbb{Z}$ two integers such that
  $F(\z) \!\!\underset{z\rightarrow 0}{\sim}\!\! a_0 \z^{k_0}$ and $F(\z)
  \!\!\!\underset{z\rightarrow \infty}{\sim}\!\!\! a_\infty
  \z^{k_\infty}\!$, with $a_0, a_\infty \in\mathbb{C}(\z) \setminus \{0\}$. These integers will be called \emph{asymptotic
    exponents of $F$} (at zero and infinity).
\end{definition}

\begin{theorem}\label{Macie}
  Let $V\in\mathbb{C}(\q_1,\q_2)\setminus \{ 0 \}$ be a homogeneous
  potential with polar representation $(F,k)$,
  let $\Lambda$ be the set defined in~\eqref{Lambda:def} (Theorem~\ref{thm:morales_polar}), counting
  multiplicities, and $k_0,k_\infty \in \mathbb{Z}$ be the asymptotic exponents of~$F$. If $k_0k_\infty\neq 0$ and if~$F$ has only
  simple Darboux points, then:
\begin{equation}\label{eqMac}
\sum\limits_{\lambda\in\Lambda}\frac{1}{\lambda-k}=\frac{1}{k_0}-\frac{1}{k_\infty}.
\end{equation}
\end{theorem}

This result is a generalization of \cite[Theorem 2.3]{MaPr05} for \emph{polynomial} homogeneous potentials, and was already proved under an equivalent form in \cite[Theorem 1.7]{StPr13}. Still, in~\cite{StPr13}, the polar representation is not used, leading to a more complicated description of the set $\Lambda$, a less readable relation \eqref{eqMac} (but not harder to compute in practice) and a more complicated proof. This is why we display here a simple self-contained proof of Theorem \ref{Macie}.

\begin{proof}
Consider the rational function $T(\z)=\z^{-2}{F(\z)}/{F'(\z)}$.
We study its poles. Since $F(\z) \!\!\underset{z\rightarrow 0}{\sim}\!\! a_0 \z^{k_0}$ with $k_0\neq 0$, the origin is a pole of $T$. Moreover, the power series expansion of $T$ at $z=0$ gives $T(\z) = 1/(k_0\z)+o(1),$ which shows that $z=0$ is a pole of order $1$ of $T$, with residue $1/k_0$. Similarly, $z=\infty$ is a pole of order~$1$ of $T$ with residue $1/k_\infty$.

Let now $z_0 \neq 0$ be a finite pole of $T$. It is either a pole of $F$, or a
root of $F'$. Any nonzero pole of $F$ of order $j$ is a pole of $F'$ of order
$j+1$, and thus it is not a pole of $T$. Thus, $z_0$ is necessarily a root
of~$F'$. By the assumption that all Darboux points of $F$ are
simple, $z_0$ is not a root of $F''$. Therefore, $z_0$ is a pole of $T$ of
order~$1$, and the series expansion
$F'(\z)=F''(z_0)(\z-z_0)+o\left((\z-z_0)^2\right)$ shows that the residue of $T$
at $z_0$ is $F(z_0)/(z_0^2 F''(z_0))$. Recognizing this expression as
$1/(k-\lambda)$ for some $\lambda\in \Lambda$, and using Cauchy's residue
formula, proves Equation~\eqref{eqMac}. 
\end{proof}

Theorem~\ref{Macie} contains two assumptions, that $k_0k_\infty\neq 0$
and that Darboux points of $F$ are simple. The first
assumption does not always hold, and then the possible eigenvalues
could be unbounded, as proven for instance by a family of potentials
in~\cite{CoKo12}, for which stronger integrability conditions were
needed (there is a similar difficulty in \cite{MaSi11}). The second
hypothesis (simple Darboux points) is not always
satisfied, but \cite[Theorem 1]{Combot13} provides a classification of
integrable potentials with a multiple Darboux point:
they are invariant by rotation, i.e. with $F$ constant.
This motivates the following definition.

\begin{definition}\label{def:Pi}
  We say that a homogeneous potential $V\in\mathbb{C}(\q_1,\q_2)$ with
  polar representation $(F,k)$ has property $\mathscr{P}$ if it
  satisfies one of the following conditions:
\begin{enumerate}
\item[(1)]\label{Pi:A} All Darboux points of $F$ are simple and the associated eigenvalues belong to the Morales-Ramis table \\ (By Theorem~\ref{thm:morales_polar}, this condition is equivalent to $\Lambda \subseteq E_k\setminus \{k\}$.)
\item[(2)]\label{Pi:B} 
$F$ is finite and nonzero either at the origin, or at infinity (i.e., $k_0k_\infty=0$).
\end{enumerate}  
\end{definition}

Remark that condition (1) includes the case $F=0$ (since then $\Lambda = \emptyset$), and that condition (2) includes the case $F$ constant nonzero.

In the case of odd homogeneity degree $k$, the function
$F$ is odd due to Proposition \ref{prop:polarV}. Thus the asymptotic exponents
$k_0,k_\infty$ are odd, and so condition (2) of
$\mathscr{P}$ cannot occur. Therefore, for odd homogeneity degrees,
$\mathscr{P}$ is equivalent to condition (1), {which matches exactly
the integrability conditions given by 
\cite[Theorem 1]{Combot13} and Theorem \ref{thm:morales_polar}.}

\subsection{Solving the Diophantine equation}\label{sec:Dio}
Equation~\eqref{eqMac} in Theorem~\ref{Macie} provides a constraint on the possible eigenvalues for a homogeneous potential $V$. 
We are thus naturally led to study the equation 
\begin{equation}\label{eq:Dio}
	\sum\limits_{i=1}^p 
	\frac{1}{\lambda_i-k}=c,
\end{equation}
where $c$ is a rational number, $k$ the homogeneity degree of $V$, and $p$  an integer related to the number of Darboux points of~$V$.

Assume that $V$ is integrable and the assumption of Theorem~\ref{Macie} are satisfied.
With our notation, Theorem~\ref{thm:morales_polar} states that $\Lambda \subseteq E_k$, where $E_k$ is the set of allowed values in the Morales-Ramis table for degree~$k$. Since the relation~\eqref{eqMac} holds, 
the aim is to solve equation~\eqref{eq:Dio} for unknowns $\lambda_1, \ldots, \lambda_p$ in $E_k$. We will prove that there are only finitely many solutions of this type.
\vspace{-0.5em}
\begin{prop}\label{prop:bound}
	For any solution $(\lambda_1, \ldots, \lambda_p)$ of the equation~\eqref{eq:Dio} the following holds
$$\min \lambda_i \leq \frac{p}{c}+k \;\hbox{ if } c>0, \quad \text{and}\quad \min \lambda_i \leq k \; \hbox{ if } c\leq 0.$$
\end{prop}

\begin{proof}
  In the case $c\leq 0$, at least one term in the sum~\eqref{eq:Dio}
  should be negative, and thus $\min \lambda_i \leq k$. Let us now
  look at the case $c>0$. Assume that we have $\lambda_i>p/c+k$ for
  all $1\leq i \leq p$. Then $(\lambda_i-k)^{-1}<c/p$ and thus
$\sum\limits_{i=1}^p \frac{1}{\lambda_i-k}<c$,
which is a contradiction with~\eqref{eq:Dio}.
This proves the proposition.
\end{proof}

Let us now remark that all entries of the Morales-Ramis table are
bounded below by $\min(0,k)$ (and this minimum is reached for $j=0$ or
$1$). Starting from this observation, we design a recursive algorithm
that finds all the solutions of equation \eqref{eq:Dio} that belong to
the Morales-Ramis table.

\smallskip

\noindent\underline{\sf MoralesRamisDiophantineSolve}\\
\textsf{Input:} The parameters $p,k,c$ of the equation~\eqref{eq:Dio}.\\
\textsf{Output:} The set of all solutions $(\lambda_1,\dots,\lambda_p)$ in $E_k$ of~\eqref{eq:Dio}, up to permutation. \\
\begin{my_enumerate}
\item If $p=1$ and $c=0$, then return $\emptyset$. If $p=1$ and $c\neq
  0$, then return $1/c+k$ if it belongs to $E_k$, else $\emptyset$. If
  $p>1$, beginning by $j=0,-1$, compute the elements of $E_k$, up to
  the bound of Proposition~\ref{prop:bound}. This yields a set $S$.
\item For each entry $S_i$ of $S$, recursively run the algorithm on
  the input $p-1,k,c-1/(S_i-k)$, with output $R_i$.
\item Return the set of solutions $[R_i,S_i],\; i=1\dots \sharp S$.
\end{my_enumerate}

Due to Proposition \ref{prop:bound}, the equation~\eqref{eq:Dio} has finitely
many solutions $(\lambda_1,\dots,\lambda_p)$ in $E_k$ (this was already proved in \cite[Lemma B.1]{MaPr05}),
and algorithm {\sf MoralesRamisDiophantineSolve} always terminates. In practice, this algorithm is
very costly. The case $k=0,c=1$ with the constraint $\lambda_i\in
\mathbb{N}$ leads to the equation analysed in \cite{Kellogg21}, for
which an optimal bound on $\max(\lambda_1,\ldots,\lambda_p)$ is found. This bound is doubly exponential
in~$p$ (which in our problem is the number of Darboux
points). It is natural to conjecture that a similar doubly exponential
bound holds in our case.

\section{The algorithm}\label{secThe}

\subsection{Specifications}

Let $\a=(\a_1, \ldots, \a_n)$ be parameters and $V\in
\Q(i)(\mathbf{a})(\q_1, \q_2)$ be a parametrized homogeneous
potential. {\em In the sequel, we assume that $V$ is given in
  canonical form, i.e. the coefficients of its numerator and
  denominator lie in $\Q[\a]$.}

Our goal is to compute a subset $\mathscr{I}(V)$ in the set of parameter values~$a$ such that $a\in \mathscr{I}(V)$ is a necessary condition for the integrability of $V(a, \q_1, \q_2)$.

In Section~\ref{sec:polar}, we have defined the polar
representation of a homogeneous potential with coefficients in
$\C$. We can do the same in the context of {\em parametrized}
homogeneous potential by defining the function 
$F(\a, \z)=V\left (\a, \frac{\z+\z^{-1}}{2}, \frac{\z-\z^{-1}}{2i}\right )$
as in the proof of Proposition \ref{prop:polarV}. With this
definition, the following lemma is an immediate consequence of Proposition \ref{prop:polarV}.

\begin{lem}\label{lemma:D}
  Let $\mathscr{D}$ be the complementary of the common solutions of
  the coefficients of the denominator of $V$. For all $a\in
  \mathscr{D}$, $(F(a, \z), k)$ is the polar representation of $V(a,
  \q_1, \q_2)$.
\end{lem}

This allows us to define the following set. We let
$\mathscr{I}(V)=\mathscr{I}(F, k)$ be the set of values
$a$ such that $a \in \mathscr{D}$ and $F(a,
\z)$ has property $\mathscr{P}$ (Definition \ref{def:Pi}).

\begin{cor}
  Let $V\in\mathbb{Q}(\mathbf{a})(\q_1,\q_2)$ be a parametrized
  homogeneous potential, and let $(F,k)$ be its polar representation,
  $F\in\mathbb{Q}(i)(\mathbf{a})(\z)$. If $V(a,\q_1, \q_2)$ is
  integrable, then ${a}\in {\mathscr{I}(F,k)}$.
\end{cor}

\begin{proof}
  Assume that $V({a},\q_1, \q_2)$ is integrable. Then thanks to
  \cite[Theorem 1]{Combot13}, if $F({a},\z)$ has a multiple
  Darboux point, then $F({a},\z)$ is constant and thus
  ${a}\in {\mathscr{I}(F,k)}$. If $F({a},\z)$ has only simple
  Darboux points, Theorem \ref{thm:morales_polar}
  implies that eigenvalues at Darboux points are
  all in $E_k$. Thus condition (1) of $\mathscr{P}$ is
  satisfied, and thus ${a}\in {\mathscr{I}(F,k)}$.
\end{proof}

Our main algorithm \textsf{IntegrabilityConditions} in
Section~\ref{ssec:generalinput} will take as input a parametrized homogeneous
potential $V\in \Q(i)(\a)(\q_1, \q_2)$ and will compute polynomial constraints
in $\Q[\a]$ that define the Zariski closure of $\mathscr{I}(V)$.

It uses a subroutine that takes as input special parametrized polar
representations $(G,k)$ and computes polynomial constraints that
define the intersection of $\mathscr{I}(G,k)$ with the subset of the
parameter space over which the valuation and number of roots/poles of
$G$ counted with multiplicities is constant.

\subsection{Subroutine for model functions} 
\vspace{-1em}

\begin{definition}\label{def:model}
We say that $G\in\mathbb{Q}(\mathbf{w})(\z)$ is a \emph{model function} if either $G=0$ identically, or there exist $\alpha\in\mathbb{Z},\;\beta_i\in\mathbb{N}$ with finitely many non zero $\beta_i$'s, such that
\begin{equation}\label{eq2}
G= \mathbf{w}_{0,0} \z^\alpha\!\!\prod\limits_{\underset{\beta_i> 0}{i\in\mathbb{Z}^*}} B_i^i
\vspace{-1em}
\end{equation}
with $B_i=\z^{\beta_i}+\sum_{j=0}^{\beta_i-1} \mathbf{w}_{i,j} \z^j$. We then denote this function by $G=G_{\alpha,\beta}$; we will write $N_{\alpha,\beta}$ for the number of parameters of $G_{\alpha,\beta}$.
\end{definition}

In the following, when there will be no ambiguity on $\alpha, \beta$ (which will be mostly the case), they will be omitted in the subscripts.

\begin{definition}\label{def:sets}
Assume $G_{\alpha, \beta}\neq 0$. We let $\Omega_{\alpha, \beta}\subset \mathbb{C}^{N_{\alpha,\beta}}$ be the subset of the parameter space defined by $\Pi(G_{\alpha,\beta})\neq 0$ where
$$\Pi(G_{\alpha,\beta})= \mathbf{w}_{0,0} \prod\limits_{\beta_i>0} \mathbf{w}_{i,0} \prod\limits_{\beta_i>0} \res(B_i,B_i') \prod\limits_{\beta_i>0,\beta_j>0,j\neq i} \res(B_i,B_j)$$
where $\res (A_1,A_2)$ denotes the resultant of two polynomials $A_1,A_2\in\mathbb{Q}[\mathbf{w}][\z]$ with respect to $\z$.
\end{definition}

Remark that for a $w\in\Omega_{\alpha, \beta}$, the roots of the $B_i$'s are all simple and non zero. Moreover, the $B_i$'s do not have any common root.

\begin{definition}\label{def:Z}
Assume $G_{\alpha, \beta}\neq 0$. Let $S$ be a finite subset of $\mathbb{Q}$. We define the polynomials in $\mathbb{Q}[\mathbf{w},\z]$
\begin{equation}\label{eq:Z2S}
Z_1=\numer\left (\frac{G'_{\alpha,\beta}}{G_{\alpha,\beta}}\right ),\; Z_{2,S}=\prod\limits_{\lambda\in S} \numer\left( k-\frac{\z^2G''_{\alpha,\beta}}{G_{\alpha,\beta}}-\lambda \right),
\end{equation}
where $\numer(f)$ denotes the numerator of $f$. 
\end{definition}

The rest of this section is devoted to the design of an algorithm called {\sf IntegrabilityConditionsModelFamily}, that takes as input a model family $G_{\alpha,\beta}$ and an integer $k$ and returns a set of polynomial equations and inequalities in $\mathbb{Q}[\mathbf{w}]$ that define the intersection of $\Omega_{\alpha, \beta}$ and $\mathscr{I}(G_{\alpha, \beta}, k)$.

We are now ready to describe our algorithm.

\smallskip
\noindent\underline{\sf IntegrabilityConditionsModelFamily}\\
\textsf{Input}: A model family $G_{\alpha,\beta}$ and an integer $k\neq -2,0,2$.\\
\textsf{Output}: {A pair $(L, \Pi(G_{\alpha, \beta}))$ such that $L$
is a list of lists of polynomials $L_1, \ldots, L_\ell$ and $\Omega_{\alpha,
  \beta}\cap \mathscr{I}(G_{\alpha, \beta},k)$ is defined by the 
union of the zero-sets of the polynomials in $L_i$ and at which
$\Pi(G_{\alpha, \beta})\neq 0$ for $1\leq i \leq \ell$.}
\begin{my_enumerate}
\item\label{model:step:0:bis} If $G_{\alpha, \beta}=0$ then return $(\emptyset, \emptyset)$. 
\item\label{model:step:ineq} Compute the polynomial $\Pi(G_{\alpha, \beta})$.
\item\label{model:step:2} If $\alpha(\alpha+\sum_{i\in \mathbb{Z}} i \beta_i)= 0$ then return $(\emptyset, \Pi(G_{\alpha, \beta}))$.
\item\label{model:step:special} Else
   \begin{my_enumerate}
\item\label{model:step:1}
Compute the coefficients $c,p$ of the relation $\sum\limits_{j=1}^p \frac{1}{\lambda_j-k}=c$ with $p=\deg_{\z} Z_1(\mathbf{w},\z)$ and $c= 1/\alpha-1/(\alpha+\sum_{i\in \mathbb{Z}} i \beta_i)$.
\item\label{model:step:3} Solve this equation using algorithm {\sf DiophantineSolve}; let $\mathscr{S}$ be its output.
\item\label{model:step:4} For each solution $S$ in $\mathscr{S}$, build the polynomial $Z_{2,S}$.
\item Compute the remainder $R_S$ for the Euclidean division of
  $Z_{2,S}$ by $Z_1$ in $\mathbb{Q}[\mathbf{w}][\z]$ and let
  $\mathscr{L}_S$ be the sequence of polynomials $ R_{i,S}$ for $i\geq 0$.
\item\label{model:step:5} Let $\mathscr{L}$ be the concatenation of
  all $\mathscr{L}_S$ for $S\in \mathscr{S}$; return $(\mathscr{L},
  \Pi(G_{\alpha,\beta}))$.
   \end{my_enumerate}
\end{my_enumerate}

Before proving the correctness of our algorithm, we will first prove the following two lemmas.

\begin{lem}\label{lem:Z1}
Assume that $G_{\alpha,\beta}\neq 0$, $\alpha\neq 0$ and let $w\in\Omega_{\alpha,\beta}$. The set of Darboux points of $G_{\alpha,\beta}(w,\z)$ is equal to the set of roots of $Z_1(w,\z)$. Moreover, if $\zeta$ is a simple Darboux point of $G_{\alpha,\beta}(w,\z)$, then $\zeta$ is a simple root of $Z_1(w,\z)$.
\end{lem}

\begin{proof}
  Let us first prove that any root of $Z_1(w,\z)$ is a
  Darboux point of $G$. Let $\zeta$ be a root of $Z_1(w,\z)$. We need
  to prove the following
$$\zeta\neq 0,\; G(w,\zeta)\neq 0,\; G'(w,\zeta)=0.$$
Consider the logarithmic derivative of $G$
$$\frac{G'}{G}(\mathbf{w},\z)= \frac{\alpha}{\z} +\sum\limits_{\beta_i>0} i\frac{B_i'}{B_i}.$$
Taking the numerator of this expression, we obtain 
\begin{equation}\label{eq:Z1}
Z_1(\mathbf{w},\z)=\alpha\prod\limits_{\beta_i>0} B_i+\z \sum \limits_{\beta_j>0} jB_j' \prod\limits_{\beta_i>0,i\neq j} B_i.
\end{equation}
Evaluating this expression at $(\mathbf{w},\z)=(w,0)$ gives
$$Z_1(w,0)=\prod\limits_{\beta_i>0} w_{i,0}.$$
This quantity is non zero as $w\in \Omega$ (due to Definition
\ref{def:sets}). Thus $\zeta\neq 0$.  Let us now prove that any non zero
root and pole of $G(w,\z)$ is not a root of $Z_1(w,\z)$. Let
$\eta\neq 0$ be a root or a pole of $G$. Then $\eta$ cancels
one and only one of the factors $B_i$ (let's say $B_{i_0}$) of $G$
(because for $w\in \Omega$, the $B_i$'s have no common root due to
Definition \ref{def:sets}). Now evaluating the expression
\eqref{eq:Z1} at $\eta$, we obtain
$$Z_1(w,\eta)=\eta i_0B_{i_0}'(w,\zeta) \prod\limits_{\beta_i>0,i\neq i_0} B_i(w,\eta).$$
This quantity is non zero because $B_{i_0}'(w,\eta)\neq 0$ (in
Definition \ref{def:sets}, the $B_i$'s have only simple roots). Thus
$\eta$ is not a root of $Z_1(w,\z)$, and therefore $\zeta$ is not a
root nor a pole of $G(w,\z)$. To conclude, we have
\begin{equation}\label{eq:Z12}
G'(\mathbf{w},\z)= G(\mathbf{w},\z) \frac{Z_1(\mathbf{w},\z)}{\z \prod\limits_{\beta_i>0} B_i(\mathbf{w},\z)}.
\end{equation}
The function $G$ is well defined at $(\mathbf{w},\z)=(w,\zeta)$
(i.e. has a finite value), the $B_i$'s do not vanish at
$(\mathbf{w},\z)=(w,\zeta)$, and thus $G'(w,\zeta)=0$.

Let us now prove the reverse. If $\zeta$ is a Darboux
point of $G(w,\z)$, then $G(w,\zeta)$ is finite and non zero, and
$G'(w,\zeta)=0$. Using equality \eqref{eq:Z12} at
$(\mathbf{w},\z)=(w,\zeta)$, we obtain $Z_1(w,\zeta)=0$.

Finally, let us look at multiplicity. If $\zeta$ is a simple
Darboux point, then $G''(w,\zeta)\neq 0$. So we
differentiate relation \eqref{eq:Z12} in $\z$ and evaluate it at
$(\mathbf{w},\z)=(w,\zeta)$. On the right-hand side, all terms vanish
except {$G(w,\zeta) \frac{Z_1'(w,\zeta)}{\zeta
    \prod\limits_{\beta_i>0} B_i(w,\zeta)}$}. As $\zeta$ is a
 Darboux point, it is neither a pole of~$G$, nor a root of any
$B_i$, and thus $Z_1'(w,\zeta)\neq 0$. Therefore, $\zeta$ is a simple
root of $Z_1(w,\z)$.
\end{proof}

\begin{lem}\label{lem:inclusion}
  Assume that $G_{\alpha,\beta}\neq 0$, $\alpha\neq 0$ and
  $\alpha+\sum_{i\in \mathbb{Z}} i \beta_i \neq 0$. Let $S$ be a
  finite set with $k \notin S$ and let
  $w\in\Omega_{\alpha,\beta}$. The polynomial $Z_1(w,\z)$ divides
  $Z_{2,S}(w,\z)$ if and only if $\Lambda(G_{\alpha, \beta}, k)
  \subset S$ and all Darboux points of
  $G_{\alpha,\beta}(w,\z)$ are simple.
\end{lem}

\begin{proof}
  Let us first assume that $Z_1(w,\z)$ divides $Z_{2,S}(w,\z)$. Thus
  all roots of $Z_1$ are roots of $Z_{2,S}$. By Lemma \ref{lem:Z1},
  the set of roots of $Z_1(w,\z)$ is the set of Darboux
  points of $G$. Let $\zeta$ be a Darboux point of
  $G$. As $Z_1(w,\zeta)=0$, we have $Z_{2,S}(w,\zeta)=0$ and thus at
  least one factor of the product defining $Z_{2,S}(\mathbf{w},\z)$
  (Eq.~\ref{eq:Z2S}) is zero. So there exists $\lambda_0\in
  S$ such that
$$\hbox{numer}\left( k-\frac{\z^2G''(\mathbf{w},\z)}{G(\mathbf{w},\z)}-\lambda_0 \right)_{(\mathbf{w},\z)=(w,\zeta)}=0.$$
We have
\begin{equation}\label{eq:eigen}
k-\frac{\z^2G''(\mathbf{w},\z)}{G(\mathbf{w},\z)}-\lambda_0=\frac{\hbox{numer}\left( k-\frac{\z^2G''(\mathbf{w},\z)}{G(\mathbf{w},\z)}-\lambda_0 \right)}{\hbox{denom}\left( k-\frac{\z^2G''(\mathbf{w},\z)}{G(\mathbf{w},\z)}-\lambda_0 \right)}.
\end{equation}
To prove that the left-hand side of this equality equals $0$ at
$(w,\zeta)$, we only need to prove that the denominator of the
right-hand side does not vanish. This denominator is always a product
of a power of $\z$, $\mathbf{w}_{0,0}$ and powers of
$B_i(\mathbf{w},\z)$. If such a product vanishes at $(w,\zeta)$, then
exactly one of the $B_i$'s vanishes (as $\zeta\neq 0$ and the $B_i$'s
do not have common roots), and then either $\zeta$ is a root or a pole
of~$G(w,\z)$. This is impossible, since $\zeta$ is a
Darboux point of $G(w,\z)$. Thus
$k-\zeta^2G''(w,\zeta)/G(w,\zeta)=\lambda_0$. So the eigenvalue
associated to $\zeta$ is $\lambda_0$ and it belongs to $S$. Thus
$\Lambda(G,k) \subset S$. Moreover, as $k\notin S$, we have $k\notin
\Lambda(G,k)$, and by Proposition \ref{prop:spectrum}, this implies
that $G''(w, \zeta)\neq 0$ and all Darboux points of
$G$ are simple.

Conversely, assume that $\Lambda(G,k) \subset S$ and all Darboux
points are simple. Let $\zeta$ be a root of $Z_1(w,\z)$. Then $\zeta$
is a Darboux point, and thus $k-\zeta^2G''(w,\zeta)/G(w,\zeta) \in
\Lambda$. Then there exists $\lambda_0\in S$ such that
$k-\zeta^2G''(w,\zeta)/G(w,\zeta)=\lambda_0$. Thus $\hbox{numer}\left(
  k-\frac{\z^2G''(\mathbf{w},\z)}{G(\mathbf{w},\z)}-\lambda_0 \right)$
evaluated at $(w,\zeta)$ equals $0$, and so $Z_{2,S}(w,\zeta)=0$. So
all roots of $Z_1(w,\z)$ are roots of $Z_{2,S}(w,\z)$. As moreover all
roots of $Z_1(w,\z)$ are simple, $Z_1(w,\z)$ divides $Z_{2,S}(w,\z)$.
\end{proof}

\vspace{-0.4cm}
\begin{theorem}\label{theo:correctness:model}
Algorithm {\sf IntegrabilityConditionsModelFamily} takes as input a model function $G_{\alpha, \beta}$ and returns a set of polynomial constraints that define $\Omega_{\alpha, \beta}\cap \mathscr{I}(G_{\alpha, \beta},k)$.
\end{theorem}

\begin{proof}
When $G$ is identically $0$, i.e. there is no Darboux point, the set of parameters is $\{\bullet\}=\mathbb{C}^0$ and the equalities $\mathscr{I}=\Omega=\{\bullet\}=\mathbb{C}^0$ hold by convention. An empty list is returned (Step \ref{model:step:0:bis}) since there is no parameter. For the rest of the proof, we may assume that $G$ is not $0$ identically.

Let us denote by $\mathscr{O}$ the set defined by the output of the algorithm {\sf IntegrabilityConditionsModelFamily}. Let us first prove that $\mathscr{O} \subset \Omega_{\alpha, \beta}\cap \mathscr{I}(G_{\alpha, \beta},k)$. First remark that $\mathscr{O}\subset\Omega$ as the output of {\sf IntegrabilityConditionsModelFamily} always contains $\Pi(G_{\alpha,\beta})\neq 0$.

\textbf{First case}: $\alpha(\alpha+\sum_{i\in \mathbb{Z}} i \beta_i)= 0$. The output returned at Step \ref{model:step:2} is simply $\Pi(G)\neq 0$, and thus $\mathscr{O}=\Omega$. Remark that when $w\in \Omega$, $0$ is not a root of the $B_i$'s. So the asymptotic exponent of $G(w,\z)$ at $0$ is $k_0=\alpha$. At infinity, the degrees of the $B_i$'s are exactly $\beta_i$ (as the polynomials are monic). Thus we obtain $k_\infty=\alpha+\sum_{i\in\mathbb{Z}} i \beta_i$, and we have $k_0k_\infty=0$, and the second case of Property $\mathscr{P}$ is satisfied. So $\mathscr{P}$ is satisfied. Thus $\mathscr{O} \subset \Omega_{\alpha, \beta}\cap \mathscr{I}(G_{\alpha, \beta},k)$.

\textbf{Second case}:  $\alpha(\alpha+\sum_{i\in \mathbb{Z}} i \beta_i)\neq 0$. Let $w\in \mathscr{O}$. Thus $w\in \Omega$ (because of the inequality $\Pi(G)\neq 0$ returned at Step \ref{model:step:5}). Moreover, there exists a $S\in \mathscr{S}$ such that $Z_1(w,\z)$ divides $Z_{2,S}(w,\z)$ (as the Step \ref{model:step:5} returns the disjunction $\bigvee_{S\subset \mathscr{S}}\mathscr{L}_S$). The list $S$ belongs to $\mathscr{S}$ thanks to Step \ref{model:step:4} and $\mathscr{S}$ is the output of {\sf DiophantineSolve} with parameters $p=\deg_{\z} Z_1(\mathbf{w},\z),k,c= 1/\alpha-1/(\alpha+\sum_{i\in \mathbb{Z}} i \beta_i)$ (Step \ref{model:step:3}). Thus $k\notin S$. The hypotheses of Lemma \ref{lem:inclusion} are now satisfied, which gives $\Lambda \subset S$ and all  Darboux points of $G$ are simple. As $S \subset E_k$, the condition $\Lambda \subset E_k$ is satisfied and then so is Property $\mathscr{P}$.

\smallskip
Conversely, let us now prove $\Omega_{\alpha, \beta}\cap \mathscr{I}(G_{\alpha, \beta},k) \subset \mathscr{O}$.

\textbf{First case}: $\alpha(\alpha+\sum_{i\in \mathbb{Z}} i \beta_i)= 0$. Then we have for the output $\mathscr{O}=\Omega$ (Step \ref{model:step:2}), and so the inclusion is trivially satisfied.

\textbf{Second case}: $\alpha(\alpha+\sum_{i\in \mathbb{Z}} i \beta_i)\neq 0$. Let $w\in \Omega_{\alpha, \beta}\cap \mathscr{I}(G_{\alpha,\beta},k)$. As $w\in \Omega$, the asymptotic exponents of $G(w,\z)$ are $k_0=\alpha,k_\infty=(\alpha+\sum_{i\in \mathbb{Z}} i\beta_i)$. Thus $k_0k_\infty\neq 0$, and so the second property of $\mathscr{P}$ is not satisfied. So the first one has to be satisfied (as $w\in \mathscr{I}(G_{\alpha, \beta},k)$). So all Darboux points of $G$ are simple, and $k_0k_\infty\neq 0$: the hypotheses of Theorem \ref{Macie} are satisfied, and thus the set $\Lambda$ satisfies the relation \eqref{eqMac}. As moreover $\Lambda \subset E_k$ (due to the first property of $\mathscr{P}$), the set $\Lambda$ should be up to permutation one (let us say $S_0$) of the lists of the output $\mathscr{S}$ of the algorithm {\sf DiophantineSolve} with parameters $c=1/k_0-1/k_\infty, p=\deg_{\z} Z_1(w,\z)$.

We now need to check that those parameters are indeed those we use in Step \ref{model:step:1}. In particular, we have yet to prove that $\deg_{\z} Z_1(w,\z)=\deg_{\z} Z_1(\mathbf{w},\z)$. We use relation \eqref{eq:Z12}. As $w\in\Omega$, the asymptotic exponent of $G(w,\z)$ at infinity is $\alpha+\sum_{i\in \mathbb{Z}} i \beta_i \neq 0$ (and equal to the asymptotic exponent of $G(\mathbf{w},\z)$). So the asymptotic exponent of $G'(w,\z)/G(w,\z)$ at infinity is $-1$. For any $w$, the degree of the $B_i$'s are $\beta_i$. So the asymptotic exponent of $Z_1(w,\z)$ (its degree) is $-1+1+\sum_{i\in\mathbb{Z}^*} \beta_i$, which is equal to the one of $Z_1(\mathbf{w},\z)$.

Thus the set $\mathscr{S}$ is indeed the one we compute in Step
\ref{model:step:3}. So in particular, we have $\Lambda \subset S_0$
(seeing now $S_0$ as a set). Now using Lemma~\ref{lem:inclusion}, we
obtain that $Z_1(w,\z)$ divides $Z_{2,S}(w,\z)$.  Step
\ref{model:step:5} returns a sequence of list of polynomials
$\mathscr{L}_1, \ldots, \mathscr{L}_\ell$ such that, for each $S\in
\mathscr{S}$ and $w$ in the solution set of some $\mathscr{L}_i$,
$Z_1(w,\z)$ divides $Z_{2,S}(w,\z)$. So this condition is
satisfied. The inequality $\Pi(G)\neq 0$ is also satisfied as
$w\in\Omega$. Thus $w$ satisfies conditions returned by {\sf
  IntegrabilityConditionsModelFamily}, and so $w\in \mathscr{O}$.
\end{proof}

\subsection{Main algorithm}\label{ssec:generalinput}

We are now ready to present our main algorithm {\sf
  IntegrabilityConditions}. It uses some basic operations on {\em
  ideals} of polynomial rings. 
 If $I,J$ are ideals,
$I:J^\infty$ denotes the saturated ideal $\{f\mid \exists g\in J, \,
\exists N\in \N,  \, fg^N\in I\}$. {We refer to \cite[Chap 2. Sec. 1]{DeckerLossen}} for more details.

\smallskip\noindent\underline{\sf IntegrabilityConditions}\\
\textsf{Input}: $V\in\mathbb{Q}(i)(\mathbf{a})(\q_1,\q_2)$ $k$-homogeneous,  given in canonical form (i.e. the coefficients of numerator / denominator of $V$ lie in $\Q[\a]$).\\
\textsf{Output}: A sequence of sets of polynomials $H_1, \ldots,
H_\ell$ in $\mathbf{a}$ such that the union of the common complex
solutions of $H_i$ defines the Zariski closure of $\mathscr{I}(V)$.
\begin{my_enumerate}
\item\label{main:step:1} Write $V$ in polar coordinates
  $V(\mathbf{a},q_1,q_2)=r^k F(\mathbf{a},e^{i\theta})$.
\item\label{main:step:1bis} Let $\Delta$ be the sequence of 
  coefficients of the denominator of $V$.
\item\label{main:step:2} Compute all the model functions $G_{\alpha,\beta}$
  whose numerator and denominator have degrees in $z$ less than those of $F$.
\item\label{main:step:3} For all possible $(\alpha, \beta)$, let
  $(\mathscr{L}_{\alpha, \beta}, W_{\alpha, \beta})$ be the output of
  {\sf IntegrabilityConditionsModelFamily} with input $G_{\alpha,
    \beta}$ and $k$.
\item Let $\mathscr{C}_{\alpha,\beta}$ be the list of coefficients in
  $\z$ of the numerator of $F-G_{\alpha,\beta}$.
\item\label{main:step:4} For each family $G_{\alpha,\beta}$ and for
  all lists $\mathscr{L}$ in $\mathscr{L}_{\alpha, \beta}$, compute a set ${H}_{\alpha,
    \beta}(\mathscr{L})$ of generators of the elimination ideal in
  $\mathbf{a}$ of the ideal generated by
$$
\langle \mathscr{L},W_{\alpha, \beta} T-1, \mathscr{C}_{\alpha,\beta} \rangle:\langle\Delta\rangle^{\infty}\cap \Q(i)[\mathbf{a}].
$$

\item\label{main:step:5} Return the list of all sets $ {H}_{\alpha,
    \beta}(\mathscr{L})$.
\end{my_enumerate}

Below, we reuse the notation $\mathscr{D}$ for the non-empty Zariski
open set introduced in Lemma \ref{lemma:D}.

\begin{theorem}
  Let $V\in \mathbb{Q}(i)(\mathbf{a})(\q_1, \q_2)$ be a $k$-homogeneous potential. Algorithm {\sf IntegrabilityConditions}
  takes as input $V$ and returns polynomial conditions defining the
  Zariski closure of $\mathscr{I}(V)$.
\end{theorem}

\begin{proof}
  In the sequel, we denote by
  $\overline{\mathscr{I}(V)}$ the Zariski closure of
  $\mathscr{I}(V)$ and we let $\mathscr{O}$ be the set defined by the
  output of Algorithm {\sf IntegrabilityConditions}. Note that since
  this set is an algebraic set, it is closed. We prove below that
  $\mathscr{O}=\overline{\mathscr{I}(V)}$.

  Take $a\in \mathscr{I}(V)$; we prove below that $a\in
  \mathscr{O}(V)$ from which we deduce that $\mathscr{I}(V)\subset
  \mathscr{O}$. Since $\mathscr{O}$ is closed for the Zariski
  topology, we conclude that $\overline{\mathscr{I}(V)}\subset
  \mathscr{O}$. Recall that, by assumption, $a\in
  \mathscr{I}(V)=\mathscr{I}(F,k)$; then $a\in \mathscr{D}$ and $F(a,
  \z)$ has property $\mathscr{P}$ (see Definition \ref{def:Pi}). We
  let $\alpha$ be the valuation of $F(a, \z)$ and $\beta_i$ be the
  number of roots/poles of multiplicity $i$. Thus, we consider the
  model function $G_{\alpha, \beta}$; for $w\in \Omega_{\alpha,
    \beta}$, $G_{\alpha, \beta}(w, \z)$ has the same features as those
  of $F(a, \z)$. This implies that $\Omega_{\alpha, \beta}$ contains a
  common root to the polynomials in $\mathscr{C}_{\alpha, \beta}$
  (Step \ref{main:step:4}) obtained after instantiating $\mathbf{a}$
  to $a$; let $w$ be one of these roots.

  Recall that $\mathscr{C}_{\alpha,\beta}$ is the list of coefficients of
  the numerator of $F-G_{\alpha, \beta}$ (seen as a polynomial in
  $\C[\z]$). Also, since $a\in \mathscr{D}$, the denominator of $F(a,
  \z)$ is not identically $0$ and we deduce that $F(a, \z)=G_{\alpha,
    \beta}(w, \z)$. Note that $w\in \Omega_{\alpha, \beta}$ by
  construction; also, by assumption, $a\in \mathscr{I}(V)$, and thus
  $F(a, \z)=G_{\alpha, \beta}(w, \z)$ has property $\mathscr{P}$. We
  deduce that $w\in \mathscr{I}(G_{\alpha, \beta}, k)$. Thus,
  correctness of Algorithm {\sf IntegrabilityConditionsModelFamily}
  (Theorem \ref{theo:correctness:model}) implies that $w$ is a
  solution of the systems output at Step \ref{main:step:3} of the main
  algorithm. Now, by construction $(w, a)$ is not a common solution of
  the polynomials in $\Delta$ and is a solution of the system obtained
  by setting to $0$ the polynomials in $\mathscr{C}_{\alpha, \beta}$,
  $\mathscr{L}$ and the inequalities $W_{\alpha, \beta}\neq 0$, for
  some $\mathscr{L}$ in $\mathscr{L}_{\alpha, \beta}$. Thus, the fact
  that $H_{\alpha, \beta}$ vanishes at $a$ is immediate and we
  conclude that $a\in \mathscr{O}$ as requested.

  To finish the proof, we establish that $\mathscr{O}\subset
  \overline{\mathscr{I}(V)}$; take $a' \in \mathscr{O}$. Then, there
  exists a set of polynomial equations $H_{\alpha,\beta}$ at Step
  \ref{main:step:4} that is satisfied by $a'$. Let $G_{\alpha, \beta}$
  be a model function associated to $(\alpha, \beta)$ (note that
  $(\alpha, \beta)$ may be non unique).

  By Theorem \ref{theo:correctness:model}, the call to {\sf
    IntegrabilityConditionsModelFamily} at Step \ref{main:step:3}
  returns a set of polynomial equations and inequalities that define
  $\Omega_{\alpha,\beta}\cap\mathscr{I}(G_{\alpha, \beta},
  k)$. Reusing the notations of Step \ref{model:step:3}, let
  $\mathscr{U}_{\alpha,\beta}$ be the constructible set defined by
  $W_{\alpha,\beta}\neq 0$, the vanishing of all polynomials in
  $\mathscr{L}$ (for some $\mathscr{L}$ in $\mathscr{L}_{\alpha,
    \beta}$) and the non-vanishing of at least one polynomial in
  $\Delta$.

  Let $\mathscr{A}_{\alpha,\beta}$ be the projection of
  $\mathscr{U}_{\alpha, \beta}$ on the $\mathbf{a}$-space; since
  $\mathscr{U}_{\alpha, \beta}$ is a constructible set,
  $\mathscr{A}_{\alpha, \beta}$ is a constructible set \cite[Chap. 3
  Sect. 2]{CLO}. By the elimination theorem \cite[Chap. 3 Theorem
  2]{CLO}, the set of equations $H_{\alpha,\beta}$ defines the Zariski
  closure of $\mathscr{A}_{\alpha,\beta}$.

\vspace{-1em}
  \begin{lem}\label{lemma:transfert}
    Let $(w, a)\in \mathscr{U}_{\alpha, \beta}$. Then $w\in
    \Omega_{\alpha, \beta}\cap \mathscr{I}(G_{\alpha, \beta}, k)$ and
    $a\in \mathscr{I}(F, k)$. 
  \end{lem}

\vspace{-1em}
  \begin{proof}
    By assumption, $w$ is a solution of the output of {\sf
      IntegrabilityConditionsModelFamily} performed with input
    $G_{\alpha, \beta}$ and $k$. The fact that $w\in \Omega_{\alpha,
      \beta}\cap\mathscr{I}_{\alpha, \beta}$ is a direct
    consequence of Theorem \ref{theo:correctness:model} that states
    the correctness of {\sf IntegrabilityConditionsModelFamily}. Since
    $w\in \Omega_{\alpha, \beta}\cap \mathscr{I}(G_{\alpha, \beta},
    k)$, we deduce that $(G_{\alpha, \beta}(w, \z), k)$ has 
    property~$\mathscr{P}$ (see Definition \ref{def:Pi}).

    Now, remark that for all $(w, a)\in \mathscr{U}_{\alpha, \beta}$,
    $F(a, \z)=G_{\alpha, \beta}(w, \z)$ (because the equations in
    $\mathscr{C}_{\alpha, \beta}$ are satisfied; see Step \ref{main:step:4}).
    We deduce that $F(a, \z)$ has the property $\mathscr{P}$; and hence
    $a\in \mathscr{I}(F, k)=\mathscr{I}(V)$.
  \end{proof}

  By Lemma \ref{lemma:transfert}, we conclude that
  $\mathscr{A}_{\alpha, \beta}\subset \mathscr{I}(F, k)$ which, by
  definition \ref{def:Pi}, is $\mathscr{I}(V)$. Then, the Zariski
  closure of $\mathscr{A}_{\alpha, \beta}$ is contained in the Zariski
  closure of $\mathscr{I}(F, k)=\mathscr{I}(V)$.  Since $\mathscr{O}$ is the
  union of the Zariski closures of $\mathscr{A}_{\alpha,\beta}$ for
  all possible $(\alpha, \beta)$, we deduce that $\mathscr{O}$ is
  contained in the Zariski closure of $\mathscr{I}(V)$. This concludes
  the proof.
\end{proof}

\vspace{-0.3cm}
\section{Experimental Results}\label{secRes}
The algorithm \textsf{IntegrabilityConditions} provides a general framework
for computing necessary conditions for integrability of homogeneous planar potentials. By incorporating more specific integrability criteria (order 2 conditions~\cite{MoRaSi07}, diagonalizability of Hessians~\cite{DuMa09}, improper Darboux points~\cite{StPr13}), the algorithm can be enhanced.
This is actually what we implemented, in the computer algebra system~\textsf{Maple (v17)}.\footnote{Our Maple implementation can be
downloaded at the url \hfill \url{http://combot.perso.math.cnrs.fr/software.html}.} The main
tools that we use involve polynomial ideals, notably relying on the Gr\"obner
engine \textsf{FGb}, implemented by J.-C. Faug\`ere~\cite{fgb} for elimination
ideal computation. This step is a crucial one for efficiency. Our
implementation succeeds in dealing with around $10$ Darboux points for simple
enough potentials (typically polynomial potentials), and around $5$ in all
cases. 

Using this implementation, we have been able to provide the first complete proof of the non-integrability of the \emph{collinear three body problem}, and even of a generalization with electrical interactions (see Theorem~\ref{th:pot3body} below). 
This implementation is also able to automatically reprove the results in ~\cite{MaPr05} about \emph{polynomial potentials\/} of degree at most~4, and, more importantly, to explore \emph{polynomial and inverses of polynomial potentials\/} with higher degree (up to 9) leading to the discovery of {\em several new candidates for integrability} that were, to our knowledge, previously unknown.


\paragraph*{Collinear three body problem}
Some classical dynamical problems, such as the collinear three body
problem (and its generalization to any homogeneity degree), can be
written as planar homogeneous potentials, of the form
\eqref{pot3body} below. For them, our algorithm is able to perform
a complete integrability analysis. Previous works on the integrability
of the three body problem treated either its simpler \emph{planar\/}
version~\cite{Boucher00,Tsygvintsev01,MoSi09}, or the collinear
version itself, but under restrictive
assumptions~\cite{Yoshida87,MoRa01,Shibayama11}.  For instance,
non-integrability was proven by Yoshida~\cite{Yoshida87} and Morales
and Ramis~\cite{MoRa01} in the case of equal masses;
Shibayama~\cite[Theorem 3]{Shibayama11} considered the case of
arbitrary masses, but his proof is valid only for the classical
collinear three body problem, and does not take into account some
exceptional cases.  Our algorithm proves the following complete
classification result. 

\begin{theorem}\label{th:pot3body}
The problem of three bodies interacting pairwise by a force proportional to the inverse of the square of the mutual distance, after reduction by translation, is not integrable.
\end{theorem}
To prove this result, we first observe that the three body problem can be rewritten as a potential of the form
\begin{equation}\label{pot3body}
V(\q_1,\q_2)={\a\q_1+\b\q_2}^{-1}+{\c\q_1+\d\q_2}^{-1}+{\e\q_1+\f\q_2}^{-1}.
\end{equation}
It is only necessary to study its integrability up to
rotation-dilatation, and thus after reparametrization, we can reduce
the problem to the following function
$F(\z)={\z}/{(\a\z^2+\b)}+{\z}/{(\c\z^2+\d)}+{\z}/{(\z^2+1)}$.
Our algorithm ({after simplification}) produces the following
integrability conditions on the parameters
\begin{align*}
[\b-\a,\d-\c],[\b-\a, \c+\d],[\a+1, \b+1],[\a+\b,\d-\c],\\
[\a+\b,\c+\d],[\a+\c,\d+\b],[\c+1,\d+1].  
\end{align*}
In all the above cases, the function $F$ simplifies to the form
$F(\z)={\z}/{(\balpha \z^2+\bbeta)}+{\z}/{(\bgamma \z^2+\bdelta)
}$. The three body problem potential is a sum of three interactions,
which have a singularity when two bodies collapse. The above functions
$F$ have only two singularities instead of three, and so at least two
bodies do not interact (as this would lead to a singularity). This
finishes the proof of Theorem~\ref{th:pot3body}. This result is here
proved for the first time for the most general form of the potential.


\noindent
\paragraph*{Polynomials and inverses of polynomials}
These potentials are very simple; they have been studied extensively
in \cite{LiMaVa11,LiMaVa11b,MaPr04,Hietarinta83,NaYo01}, and
contain some interesting integrable potentials. For our algorithm,
these potentials are also simpler than typical rational ones, because
they do not involve simplifications between numerators and
denominators. Thus fewer functions $G_{\alpha,\beta}$ need to be
analyzed. Below we reproduce some results of non-integrability of
homogeneous polynomial potentials and inverses of homogeneous
polynomial potentials, and we extend them to higher degrees than previously known thanks to
our algorithm (the positive degree corresponding to a polynomial, and
a negative degree to an inverse of a polynomial). A polynomial homogeneous potential leads to a function $F$ of the form $F(\z)=\sum_{i=0}^k a_i \z^{k-2i}$ (respectively $F(\z)=1/\sum_{i=0}^k a_i \z^{k-2i}$ for the inverse of a homogeneous polynomial). Below we give the sets of eigenvalues for candidates for integrability, and corresponding computation timings.

{\small\begin{center}
\begin{tabular}{|c||c|c|c|c|}\hline
$k$ &Eigenvalues (the set $\Lambda(F,k)$ in Eq.~\eqref{Lambda:def}) & timings\\\hline\hline
3& $\{0,6\},\{1,15\},\{45,\frac{3}{8}\},\{1,10,45\}$ &0.43s\\\hline
4&$\{0,12\},\{24,\frac{3}{2}\},\{12,84,\frac{3}{2}\},\{544,\frac{3}{2},\frac{35}{2}\}$&0.91s\\\hline
5&$\{\frac{27}{8},135\},\{\frac{7}{8},35,170,665\},\{0,20\},\{2,35\}$&433s\\\hline
6&$\{0,30\},\{48,\frac{5}{2}\}$&2.5s\\\hline
7&$\{0,42\},\{63,3\}$&33.5s\\\hline
9& $\{4\},\{0,72\},\{4,99\}$ &106s\\\hline
$-3$&$\{-2\},\{0\},\{3\},\{3,7,12\}$&0.36s\\\hline
$-4$&$\{8,20\},\{0\}$&0.639s\\\hline
$-5$&$ \{130, -\frac{13}{8}\}, \{-3, 15\}, \{22, -\frac{13}{8}\}, \{0\}, \{15\}$&3.167s\\\hline
$-6$&$\{24\},\{0\}$&594s\\\hline
\end{tabular}\\
{\small{Integrability analysis of homogeneous polynomials \\and inverses of homogeneous polynomials}}
\end{center}}
	
For polynomials of degree $3$ and $4$, we retrieve known results (leading
to a complete classification of integrable homogeneous potentials of
degree $3$, and almost complete for degree $4$). For $k=4$, we obtain after
simplification the following ideals
\begin{small}\begin{equation*}\begin{split}
[\a_1,\a_2],[\a_4,\a_5],[36\,\a_5\a_1-\a_3^2,6\a_4\a_1-\a_3\a_2,6\a_2\a_5-\a_4\a_3],[44979\a_2^2-\\
376712\a_3\a_1,66879684\a_5\a_1-75625\a_3^2,16719921\,\a_{{4}}\a_{{2}}-4708900\,{\a_{{3}}}^{2},\\
-376712\,\a_3\a_5+44979\,{\a_{{4}}}^{2},8178\,\a_{{4}}\a_{{1}}-275\,\a_{{3}}\a_{{2}},8178\,\a_{{2}}\a_{{5}}-275\,\a_{{4}}\a_{{3}}] ,\\
[-392\,\a_{{3}}\a_{{1}}+99\,{\a_{{2}}}^{2},484\,\a_{{5}}\a_{{1}}-{\a_{{3}}}^{2},1089\, \a_{{4}}\a_{{2}}-196\,{\a_{{3}}}^{2},\\
-392\,\a_{{3}}\a_{{5}}+99\,{\a_{{4}}}^{2},22\,\a_{{4}}\a_{{1}}-\a_{{3}}\a_{{2}},22\,\a_{{2}}\a_{{5}}-\a_{{4}}\a_{{3}}] ,\\
[-40\,\a_{{3}}\a_{{1}}+7\,{\a_{{2}}}^{2},15876\,\a_{{5}}\a_{{1}}-25\,{\a_{{3}}}^{2},441\,\a_{{4}}\a_{{2}}-100\,{\a_{{3}}}^{2},\\
-40\a_3\a_5+7\a_4^2,126\,\a_{{4}}\a_{{1}}-5\,\a_{{3}}\a_{{2}},126\,\a_{{2}}\a_{{5}}-5\,\a_{{4}}\a_{{3}}]
\end{split}\end{equation*}\end{small}

The first two cases are the exceptional ones with $k_0(F) k_\infty(F)= 0$, and the other cases lead indeed to integrable
potentials. These are exactly the conditions found in \cite{MaPr05}. 
At degree~$5$, {\em two non-trivial new potentials (up to
conjugation and rotation-dilatation) are detected, not known to be
integrable}, but satisfying all integrability conditions. Their
eigenvalue sets are $\{27/8,135\},\{7/8,35,170,665\}$. The second has
algebraic coefficients of degree $12$, and the first is
\begin{equation*}\begin{split}
V(\q_1,\q_2)=(\q_1+i\q_2)^2(11391716i\q_2^3+73950132\q_1\q_2^2-\\
150075213i\q_1^2\q_2-96733564\q_1^3).
\end{split}\end{equation*}
At degree $6$ and $7$, the only possible cases are either already
known, or $k_0(F) k_\infty(F)= 0$, or they do
not have Darboux points. Thus for degree $7$, the only cases whose
integrability status is still unknown are up to rotation-dilatation
$F(\z)=\z$ and $F(\z)=\z^3$. So we obtained a classification of integrable homogeneous polynomial potentials of degree~$7$. 
For degree $8$ and $9$, some optimizations are
necessary for the algorithm to be workable. Indeed, thanks to the fact that our family is invariant by rotation-dilatation, it is only necessary to consider functions $G_{\alpha,\beta}$ with the coefficient $w_{0,0}=1$ and with the trailing coefficient of one polynomial factor equal to $1$. This removes two variables in the elimination ideal, and reduces by $2$ the Hilbert dimension of the output. 
At degree~$9$, we find {\em three new cases satisfying all integrability conditions}; they are given by
$F(\z)=\z+\z^{-5},\;\;F(\z)=\z^3+\z^{-3},\;\;F(\z)=\z^{-5}(\z^2+1)^5$.

\bibliographystyle{abbrv}
\vspace{2em}
\begin{scriptsize}
	\vspace{2em}

\def\cprime{$'$}

\end{scriptsize}

\end{document}